\documentclass{article}

\usepackage{graphicx} 
\usepackage{amsmath} 
\usepackage{amssymb}  
\usepackage{amsthm}

\renewcommand{\c}[1]{^{(#1)}}
\newcommand{\ext}{^{(0)}}
\newcommand{\diffparam}{k_c}
\newcommand{\degparam}{k_d}

\newcommand{\T}{^\mathrm{T}}
\newcommand{\Real}{\mathbb{R}}

\renewcommand{\operatorname}[1]{\mathop{\mathrm{#1}}}
\renewcommand{\det}{\operatorname{det}}

\newcommand{\normx}{\Vert x \Vert}
\newcommand{\normy}{\Vert y \Vert}

\newtheorem{lemma}{Lemma}
\newtheorem{theorem}{Theorem}
\newtheorem{prop}{Proposition}
\newtheorem{assum}{Assumption}

\title{\LARGE \bf
Network-level dynamics of diffusively coupled cells%
\thanks{This work was supported by the Cluster of Excellence in Simulation Technology at the University of Stuttgart and by grant WA 2965/1-1 from the German research foundation (DFG).}
}

\author{Steffen Waldherr and Frank Allg\"ower
\thanks{The authors are with the Institute for Systems Theory and Automatic Control,
        University of Stuttgart, 70550 Stuttgart, Germany
        {\tt\small waldherr@ist.uni-stuttgart.de, allgower@ist.uni-stuttgart.de}}%
}

\date{Published in Proceedings of the 51st IEEE Conference on Decision and Control (CDC), 2012, pp.\ 5517-5522. (C) IEEE.}

\begin{document}

\maketitle
\thispagestyle{empty}
\pagestyle{empty}

\begin{abstract}
  We study molecular dynamics within populations of diffusively coupled cells under the assumption of fast diffusive exchange.
  As a technical tool, we propose conditions on boundedness and ultimate boundedness for systems with a singular perturbation, which extend the classical asymptotic stability results for singularly perturbed systems.
  Based on these results, we show that with common models of intracellular dynamics, the cell population is coordinated in the sense that all cells converge close to a common equilibrium point.

  We then study a more specific example of coupled cells which behave as bistable switches, where the intracellular dynamics are such that cells may be in one of two equilibrium points.
  Here, we find that the whole population is bistable in the sense that it converges to a population state where either all cells are close to the one equilibrium point, or all cells are close to the other equilibrium point.
  Finally, we discuss applications of these results for the robustness of cellular decision making in coupled populations.

\end{abstract}

\section{INTRODUCTION}

Coordinated behavior in cell populations is a long standing topic in the analysis of biological systems.
While so far mainly the synchronization of biological oscillators has been in the focus of this research \cite{MirolloStr1990,StanHam2007,LangMar2011}, coordinated behavior is also of interest for other types of cellular behavior.
The study presented in this paper is motivated by the analysis of biological switches, which are a common phenomenon in biological systems such as gene regulation or signal transduction \cite{FerrellXio2001,EissingWal2007}.
In contrast to oscillators, synchronization among biological switches has hardly been studied so far.
Wang and coauthors \cite{WangLu2011} present a model of a genetic switch in a cell population where one of the molecules making up the switch is diffusively exchanged among cells via the extracellular medium.
However, \cite{WangLu2011} is concerned more with a numerical study of the effects of noise on synchronized switching than the underlying mechanism of switch synchronization per se.
A more theoretically oriented study in this area has been performed by Schmidt and coauthors \cite{SchmidtWu2010}.
There, conditions for consensus among directly coupled, scalar bistable switches are proven using Lyapunov techniques.
While no specific biological model is studied in \cite{SchmidtWu2010}, the analysis is motivated from the occurence of bistable dynamics in biological systems.

Here, we apply singular perturbation theory to study populations of diffusively coupled cells for diffusion rates which are significantly faster than the intracellular dynamics.
Although there are long-standing results for the convergence analysis of singularly perturbed system \cite{SaberiKha1984,CorlessGli1992}, these studies considered only systems where the perturbation does not affect the position of the equilibrium point, and then proceeded to show asymptotic or exponential stability of this equilibrium point for a sufficiently small perturbation.
In the model we are going to construct here, the equilibrium point will in fact depend on the perturbation, and the previous results are not applicable.
We can however prove a new result on boundedness and ultimate boundedness of singularly perturbed systems, which in general is applicable to cells coupled via diffusive exchange of molecules with a common extracellular medium.
As a specific application, we consider a population of coupled bistable switches.
The theoretical results and a simulation study show coordinated behavior and synchronized switching among the coupled switches for fast enough diffusive exchange rates.

The paper is structured as follows.
In Section~\ref{sec:bound-singular-perturbation}, we present conditions on boundedness and ultimate boundedness of singularly perturbed systems.
In Section~\ref{sec:cp-dynamics}, we construct a generic model for a population of diffusively coupled cells.
The previous conditions are then applied to show convergence to the neighborhood of a common equilibrium point for all individual cells.
Finally, in Section~\ref{sec:simulation}, we consider a population of diffusively coupled bistable switches, and show that the population behaves as a single switch with all cells in coordination for fast enough diffusive exchange.

\section{BOUNDEDNESS IN SINGULARLY PERTURBED SYSTEMS}
\label{sec:bound-singular-perturbation}

\subsection{Conditions for ultimate boundedness}

The analysis of the coupled switches with singular perturbation requires an extension of singular perturbation theory.
Previous singular perturbation results either assumed that the equilibrium point does not vary with the small parameter $\varepsilon$ \cite{SaberiKha1984}, or that the dependence on $\varepsilon$ of the dynamics tends to zero as the state approaches the equilibrium point \cite{CorlessGli1992}, and then proceeded to show exponential stability for this equilibrium point.
In this study, these assumptions are not satisfied and it will not be possible to show exponential stability of a fixed equilibrium point.
Instead, we will show convergence of the singularly perturbed system to a neighborhood of the equilibrium state.
By decreasing the small parameter $\varepsilon$, this neighborhood can be made arbitrarily small, thus recovering the classical result in the limit.

Consider the system
\begin{equation}
  \label{eq:orig-system}
  \begin{aligned}
    \dot x &= f(x,z) \\
    \varepsilon \dot z &= g(x,z,\varepsilon)
  \end{aligned}
\end{equation}
with $x \in D_x \subset \Real^n$, $D_x$ convex, $z \in D_z \subset \Real^m$, $\varepsilon > 0$, and functions $f$ and $g$ continuously differentiable in all of their arguments and such that a unique solution $x(t) \in D_x$, $z(t) \in D_z$ exists for all $t \geq 0$.
Initial conditions are taken as $x(0) = x_0$ and $z(0) = z_0$.

Our first assumption is that a quasi-steady-state solution for the fast dynamics can be obtained, and that the system \eqref{eq:orig-system} has an equilibrium point at $\varepsilon = 0$.
\begin{assum}
There exists a continously differentiable function $h: D_x \rightarrow D_z$ such that
\begin{equation}
  \label{eq:qss}
  g(x,h(x),0) = 0
\end{equation}
for all $x \in D_x$,
and $h$ satisfies the linear growth bound
\begin{equation}
  \label{eq:h-bound}
  \Vert \frac{\partial h}{\partial x}(x) \Vert \leq c_1
\end{equation}
for a positive constant $c_1$ for all $x \in D_x$.
Furthermore, for $\varepsilon = 0$, system \eqref{eq:orig-system} has an equilibrium point at $(0, h(0))$, i.e.,
\begin{equation}
\label{eq:ss-eq}
\begin{aligned}
  f(0,h(0)) &= 0 \\
  g(0,h(0),0) &= 0.
\end{aligned}
\end{equation}
\end{assum}

For the next step, the reduced order systems for $\varepsilon = 0$ are defined: the slow system for the variable $x$ and the boundary-layer or fast system for the variable
\begin{equation}
\label{eq:fast-variable}
y(\tau) = z(\tau) - h(x_0)
\end{equation}
on the fast time scale $\tau = t / \varepsilon$.
The slow system is defined as
\begin{equation}
  \label{eq:slow-system}
  \dot x = f(x,h(x)),
\end{equation}
and the boundary-layer system as
\begin{equation}
  \label{eq:fast-system}
  \frac{dy}{d\tau} = g(x_0, y+h(x_0), 0)
\end{equation}
where $x_0$ is considered as a constant parameter.
Also, define
\begin{equation}
  \label{eq:delta-f}
  \Delta_f(x,z) = f(x,z) - f(x,h(x))
\end{equation}
and
\begin{equation}
  \label{eq:delta-g}
  \Delta_g(x,z,\varepsilon) = g(x,z,\varepsilon) - g(x,z,0)
\end{equation}

We have the following assumptions on the reduced order systems and the functions $f$ and $g$:
\begin{assum}
  There are a scalar function $V_S(x)$ and positive
  constants $c_2$, $d_1$, $d_2$, and $D$ such that
  \begin{equation}
    \label{eq:VS-constraint}
    d_1 \Vert x \Vert^2 \leq V_S(x) \leq d_2 \Vert x \Vert^2
  \end{equation}
  and
  \begin{equation}
    \label{eq:lyap1}
    \frac{\partial V_S}{\partial x} f(x,h(x)) \leq - c_2 \Vert x \Vert^2
  \end{equation}
  for all $x \in \Omega_x$, a compact subset of $D_x$, with $\{ d_1 \Vert x \Vert^2 \leq D \} \subset \Omega_x$.
\end{assum}

\begin{assum}
  There are a scalar function $V_B(x,y)$ and positive constants $c_3$, $d_3$,
  and $d_4$, such that
  \begin{equation}
    \label{eq:VB-constraint}
    d_3 \Vert y \Vert^2 \leq V_B(x,y) \leq d_4 \Vert y \Vert^2
  \end{equation}
  and
  \begin{equation}
    \label{eq:lyap2}
    \frac{\partial V_B}{\partial y} g(x,y+h(x),0)  \leq - c_3 \Vert y \Vert^2
  \end{equation}
  for all $x \in \Omega_x$, $y$ such that $y+h(x) \in D_z$.
\end{assum}

\begin{assum}
  There exist constants $c_4$, $c_5$, $c_6$, $c_7$, $c_8$, $c_9$, $c_{10}$, $c_{11}$, and $c_{12}$,
  such that
  \begin{equation}
    \label{eq:lyap3-8}
    \begin{aligned}
      \frac{\partial V_S}{\partial x} \Delta_f(x,y+h(x)) &\leq c_4 \Vert x \Vert \Vert y \Vert \\
      \frac{\partial V_B}{\partial x} f(x,y+h(x)) &\leq c_5 \Vert y \Vert^2 + c_6 \Vert x \Vert \Vert y \Vert \\
      \frac{\partial V_B}{\partial y} \Delta_g(x,y+h(x),\varepsilon) &\leq \varepsilon ( c_7 \Vert y \Vert^2 + c_8 \Vert x \Vert \Vert y \Vert + c_9 \Vert y \Vert) \\
      \vert \frac{\partial V_B}{\partial y} \frac{\partial h}{\partial x} f(x,y+h(x)) \vert &\leq c_{10} \Vert y \Vert^2 + c_{11} \Vert x \Vert \Vert y \Vert + c_{12} \Vert y \Vert
    \end{aligned}
  \end{equation}
  for all $x \in D_x$, $y$ such that $y+h(x) \in D_z$.
\end{assum}

These assumptions are equivalent to the assumptions in the classical literature, just that here the terms $c_9 \Vert y \Vert$ and $c_{12} \Vert y \Vert$ are added.
Referring to the discussion at the beginning of this section, these terms are relevant for cases where $g$ and $f$ are not sector bounded in $\Vert x \Vert$ and $\Vert y \Vert$, but may have a bounded offset in $\varepsilon$.

First, we show that $x$ will remain within $\Omega_x$ in the perturbed system, so that \eqref{eq:lyap1} can be used later on.
\begin{lemma}
  \label{lemma:omega_x}
  Under assumptions 1--4 above, there exists $\varepsilon^{\prime}$ such that, for all $\varepsilon < \varepsilon^{\prime}$, the solution $x(t)$ of \eqref{eq:orig-system} stays in $\Omega_x$ for all $t \geq 0$ and any initial condition where $x_0 \in \{ d_2 \Vert x \Vert^2 \leq \varrho D \}$, with $0 \leq \varrho < 1$.
\end{lemma}

\begin{proof}
  From the classical results on singular perturbation \cite[Th.\ 11.2]{Khalil2002}, it can be established that the solution of \eqref{eq:orig-system} satisfies $\Vert y(t) \Vert \leq \gamma_1 e^{- \frac{\gamma_2 t}{\varepsilon}} + \gamma_3 \varepsilon$ with some positive constants $\gamma_{1,2,3}$, for $t \geq 0$.
Considering $V_S$ as a Lyapunov function for \eqref{eq:orig-system}, we find
\begin{equation}
  \label{eq:VS-orig-derivative}
  \begin{aligned}
    \dot V_S &= \frac{\partial V_S}{\partial x} (f(x,h(x)) + \Delta_f(x,y+h(x))) \\
    &\leq - c_2 \Vert x \Vert^2 + c_4 \Vert x \Vert (\gamma_1 e^{- \frac{\gamma_2 t}{\varepsilon}} + \gamma_3 \varepsilon).
  \end{aligned}
\end{equation}
Thus there exists $T > 0$, such that $x(t) \in \{ d_2 \Vert x \Vert^2 \leq D\}$ for $t \in [0,T]$ and sufficiently small $\varepsilon$.
For $t \geq T$, we have $\dot V_S \leq - c_2 \Vert x \Vert^2 + c_4 \Vert x \Vert \gamma_4 \varepsilon$ and thus $\dot V_S < 0$ whenever $d_2 \Vert x \Vert^2 = D$, for sufficienly small $\varepsilon$.
This implies that $x(t) \in \Omega_x$ for all $t \geq 0$.
\end{proof}

\begin{theorem}
  \label{theorem:boundedness}
  Suppose that assumptions 1--4 above hold and let $\delta$, $\gamma$ such that $0 <\gamma, \delta < 1$.
  Define $w(t)\T = (x(t), z(t)-h(0))\T$, where $(x(t), z(t))$ is a solution of system \eqref{eq:orig-system}.
  Let $x_0 \in \{ d_2 \Vert x \Vert^2 \leq \varrho D \}$, with $0 \leq \varrho < 1$.
  Then there exists $\varepsilon^\ast(\gamma,\delta) > 0$ such that for any $\varepsilon < \varepsilon^\ast(\gamma,\delta)$, 
  \begin{equation}
    \label{eq:bounds}
    \begin{aligned}
      \Vert w(t) \Vert^2 \leq  \beta(\Vert w(0) \Vert^2,t) + \mu(\varepsilon),
    \end{aligned}
  \end{equation}
  where $\beta$ is a class $\mathcal{KL}$ function, i.e., strictly increasing in the first argument, strictly decreasing in the second one, and $\beta(r^2,\infty)=\beta(0,t)=0$, and $\mu(\varepsilon) \rightarrow 0$ as $\varepsilon \rightarrow 0$.
\end{theorem}

\begin{proof}
We first construct a transformed system by letting $y = z - h(x)$, and denote the state of the transformed system as $v\T = (x, y)\T$.
From \eqref{eq:orig-system}, the transformed dynamics are
\begin{equation}
  \label{eq:transformed-system}
  \begin{aligned}
    \dot x &= f(x, y+h(x)) \\
    \varepsilon \dot y &= g(x,y+h(x),\varepsilon) - \varepsilon \frac{\partial h}{\partial x} f(x,y+h(x)).
  \end{aligned}
\end{equation}

For the transformed system, consider the Lyapunov function candidate
  \begin{equation}
    \label{eq:lyap-candidate}
    V(x,y) = (1-\delta) V_S(x) + \delta V_B(x,y).
  \end{equation}

Let us first check that $V$ satisfies quadratic norm constraints from below and above.
In fact, from \eqref{eq:VS-constraint} and \eqref{eq:VB-constraint}, we directly have
\begin{equation*}
  (1-\delta) d_1 \Vert x \Vert^2 + \delta d_3 \Vert y \Vert^2 \leq V(x,y)
\end{equation*}
and
\begin{equation*}
  V(x,y) \leq (1-\delta)d_2\Vert x \Vert^2 + \delta d_4 \Vert y \Vert^2.
\end{equation*}

Next, let us compute the derivative of $V(x,y)$ along trajectories of the system \eqref{eq:transformed-system}.
Using the bounds in Assumptions 2--4 and Lemma \ref{lemma:omega_x}, we obtain
\begin{equation*}
  \begin{aligned}
    \dot V(x,y) &\leq -(1-\delta)c_2 \normx^2 + (1-\delta) c_4 \normx \normy + \delta c_5 \normy^2 \\
    &\quad+ \delta c_6 \normx \normy - \frac{\delta}{\varepsilon} c_3 \normy^2 + \delta c_7 \normy^2 + \delta c_8 \normx \normy \\
    &\quad+ \delta c_9 \normy + \varepsilon \delta c_{10} \normy^2 + \varepsilon \delta c_{11} \normx \normy + \varepsilon \delta c_{12} \normy \\
    &\leq - (\normx, \normy) A(\gamma,\delta,\varepsilon) (\normx, \normy)\T \\
    &\quad + \delta B(x, y, \gamma, \varepsilon)\\
  \end{aligned}
\end{equation*}
with
\begin{small}
  \begin{equation*}
    A(\gamma,\delta,\varepsilon) =
    \begin{pmatrix}
      (1-\delta(1-\gamma))c_2 & -\frac{c_4 + \delta(c_6+c_8-c_4+\varepsilon c_{11})}{2} \\
      -\frac{c_4 + \delta(c_6+c_8-c_4+\varepsilon c_{11})}{2} &
      \frac{\delta c_3 (1-\gamma)}{\varepsilon} - \delta (c_7 +
      \varepsilon c_{10})
    \end{pmatrix}
  \end{equation*}
\end{small}
and
\begin{equation*}
  B(x, y, \gamma, \varepsilon) = (c_9 + \varepsilon c_{12}) \normy - \frac{c_3 \gamma}{\varepsilon} \normy^2 - \gamma c_2 \normx^2
\end{equation*}

Clearly, if $\delta$ and $\gamma$ satisfy the conditions in Theorem~\ref{theorem:boundedness}, we can find $\varepsilon^{\prime\prime}(\delta,\gamma) > 0$ such that $A(\gamma,\delta,\varepsilon)$ is positive definite for all positive $\varepsilon < \varepsilon^{\prime\prime}(\gamma,\delta)$.
By considering the determinant and trace of $A$, one can even compute suitable $\varepsilon^{\prime\prime}$, but the result is a bit longish and skipped here.

From Lemma~\ref{lemma:B-negative} in the Appendix, we conclude that $B(x,y,\gamma,\varepsilon) \leq 0$ whenever $\Vert v \Vert^2 = \normx^2 + \normy^2 \geq R(\varepsilon)$, with $R(\varepsilon) \rightarrow 0$ as $\varepsilon \rightarrow 0$.

Thus, from results on boundedness and ultimate boundedness \cite{CorlessLei1981,Khalil2002}, we can conclude that solutions of \eqref{eq:transformed-system} are bounded by
\begin{equation}
  \label{eq:trafo-sys-bounds}
  \Vert v(t) \Vert^2 \leq \tilde\beta(\Vert v(0) \Vert^2,t) + \tilde\mu(\varepsilon),
\end{equation}
where $\tilde\beta$ is a class $\mathcal{KL}$ function.
From Lemma~\ref{lemma:B-negative}, we have $\mu(\varepsilon) \rightarrow 0$ as $\varepsilon \rightarrow 0$.

Finally, let us transform the bound \eqref{eq:trafo-sys-bounds} back to the original variable $w\T = (x,z-h(0))\T$.
From assumption 1 and a Taylor expansion of $h(x)-h(0)$, with convexity of $D_x$, we have
\begin{equation*}
  \begin{aligned}
    \Vert w \Vert^2 &\leq \normx^2 + \Vert z - h(x) \Vert^2 + \Vert h(x) - h(0) \Vert^2 \\
    &\leq \normx^2 + \normy^2 + c_1^2 \normx^2 \\
    &\leq (1+c_1^2) \Vert v \Vert^2.
  \end{aligned}
\end{equation*}
Similarly,
\begin{equation*}
  \begin{aligned}
    \Vert v \Vert^2 &\leq \normx^2 + \Vert z - h(0) \Vert^2 + \Vert h(0) - h(x) \Vert^2 \\
    &\leq (1+c_1^2) \Vert w \Vert^2.
  \end{aligned}
\end{equation*}
Also, note that $a \leq b$ implies that $\tilde\beta(a,t) \leq \tilde\beta(b,t)$.
Thus, we conclude
\begin{equation}
  \label{eq:final-w-bound}
  \Vert w(t) \Vert^2 \leq (1+c_1^2) \bigl(\tilde\beta((1+c_1^2)\Vert w(0) \Vert^2,t) + \tilde\mu(\varepsilon)\bigr).
\end{equation}
\end{proof}

\section{DYNAMICS OF A COUPLED CELL POPULATION}
\label{sec:cp-dynamics}

\subsection{Modeling}
\label{sec:cp-modeling}

We consider a population of $N$ cells which exchange molecules with a joint extracellular medium via diffusion.
The concentrations of these molecules within cell $i$ are denoted by the vector $\xi\c{i} \in \Real^n_+$.
The intracellular dynamics, neglecting coupling effects, are assumed to be given by the differential equation
\begin{equation}
  \label{eq:cellmodel}
  \dot \xi\c{i} = F(\xi\c{i}),
\end{equation}
The vector field $F(\xi)$ is assumed to be continously differentiable.

Next, let us turn to the diffusive coupling.
The concentration of the molecules in the extracellular medium is denoted by $\xi\ext \in \Real^n_+$.
Diffusive exchange between cells and the extracellular medium is described by a flux rate which is linear in the concentration difference $\xi\c{i} - \xi\ext$, with a diffusion rate constant $\diffparam \geq 0$.
Due to volume differences between the intracellular and the extracellular compartments, a scaling factor needs to be introduced.
We will assume that all cells have identical volume $V_c$ and that the volume of the medium is $N V_c$, i.e., each cell shares a fixed part of the medium equal to its own volume.
In addition, we include a linear removal rate for the considered molecules from the extracellular medium, for example by degradation or an outflow of medium, with a rate constant $\degparam \geq 0$.

The overall dynamics of the coupled population are then given by the following system of differential equations.
\begin{equation}
  \label{eq:popmodel}
  \begin{aligned}
    \dot \xi\ext &= - \degparam \xi\ext + \sum_{j=1}^N \frac{\diffparam}{N V_c} (\xi\c{j} - \xi\ext) \\
    \dot \xi\c{i} &= F(\xi\c{i}) - \frac{\diffparam}{V_c} (\xi\c{i} - \xi\ext), \quad i=1,\dotsc,N.
  \end{aligned}
\end{equation}

\subsection{Analysis via singular perturbation theory}
\label{sec:singpert-popmodel}


We are going to analyse the population model \eqref{eq:popmodel} with the theory developed in Section~\ref{sec:bound-singular-perturbation}.
Since we are interested in the limiting case of fast diffusive exchange between cells and the extracellular medium, a reasonable small parameter is $\varepsilon = \diffparam^{-1}$.
First, we have to bring the model into the standard form given in \eqref{eq:orig-system}.

To this end, consider the variable
\begin{equation}
  \label{eq:trafo}
  x = \xi\ext + \frac{1}{N} \sum_{j=1}^N \xi\c{j},
\end{equation}
which from \eqref{eq:popmodel} is subject to the dynamics
\begin{equation*}
  \dot x = -\degparam \xi\ext + \frac{1}{N} \sum_{j=1}^N F(\xi\c{j}).
\end{equation*}
In order to have the notation as in Section~\ref{sec:bound-singular-perturbation}, we also denote $z\c{i} = \xi\c{i}$, $i=1,\dotsc,N$.
Thus, the transformed model with variables $x$ and $z$ is
\begin{equation}
  \label{eq:popmodel-transformed}
  \begin{aligned}
    \dot x &= -\degparam (x - \frac{1}{N} \sum_{j=1}^N z\c{j}) + \frac{1}{N} \sum_{j=1}^N F(z\c{j}) \\
    \varepsilon \dot z\c{i} &= \varepsilon F(z\c{i}) - \frac{1}{V_c}\bigl(z\c{i} - x + \frac{1}{N} \sum_{j=1}^N z\c{j}\bigr),
  \end{aligned}
\end{equation}
which is in the standard form for singular perturbation analysis with small parameter $\varepsilon = \diffparam^{-1}$.

To simplify notation, we introduce $f$ and $g$ such that \eqref{eq:popmodel-transformed} writes as
\begin{equation}
  \label{eq:fg-cp}
  \begin{aligned}
    \dot x &= f(x,z) \\
    \varepsilon \dot z &= g(x,z,\varepsilon).
  \end{aligned}
\end{equation}

Next, we derive the slow and boundary-layer models for the population model \eqref{eq:popmodel-transformed} according to \eqref{eq:slow-system} and \eqref{eq:fast-system}.
The first step is to solve the equation $g(x,z,0) = 0$ to obtain the quasi-steady state solution $z = h(x)$.
From the model~\eqref{eq:popmodel-transformed}, this yields the system of $n N$ linear algebraic equations
\begin{equation}
  \label{eq:limit-algebraic}
  z\c{i} - x + \frac{1}{N} \sum_{j=1}^N z\c{j} = 0, \quad i = 1,\dotsc,N,
\end{equation}
which has to be solved for $z\c{i}$ in order to obtain $h$.
One can verify directly that one solution of \eqref{eq:limit-algebraic} is
\begin{equation}
  \label{eq:singular-manifold}
  z\c{i} = \frac{x}{2},\quad i=1,\dotsc,N.
\end{equation}
To show that \eqref{eq:singular-manifold} is the unique solution, we have the following result.

\begin{lemma}
  \label{lem:algebraic-solution}
  The matrix $E_N = \frac{1}{N} \underline{\mathbf{1}} + I \in \Real^{N\times N}$, where $\underline{\mathbf{1}}$ is a $N \times N$ matrix of all ones, has one eigenvalue at $2$ and $N-1$ eigenvalues at $1$.
\end{lemma}
\begin{proof}
  First, let us compute the eigenvalues $\tilde\lambda$ of $\frac{1}{N} \underline{\mathbf{1}}$.
  From
  \begin{equation*}
    \frac{1}{N} \underline{\mathbf{1}} \mathbf{1} = \mathbf{1},
  \end{equation*}
  where $\mathbf{1} \in \Real^N$ is the vector of all ones, we can verify that one eigenvalue is $\tilde\lambda_1 = 1$.
  Since $\frac{1}{N} \underline{\mathbf{1}}$ is a rank one matrix, the other eigenvalues are $\tilde\lambda_{2\ldots N} = 0$.

  The eigenvalues $\lambda$ of $E_N$ are computed by
  \begin{equation*}
    \det((\lambda - 1) I - \frac{1}{N} \underline{\mathbf{1}}) = 0,
  \end{equation*}
  yielding
  \begin{equation*}
    \begin{aligned}
      \lambda_1 &= \tilde\lambda_1 +1 = 2 \\
      \lambda_{2\ldots N} &= \tilde\lambda_{2\ldots N} + 1 = 1.
    \end{aligned}
  \end{equation*}
\end{proof}

Note that \eqref{eq:limit-algebraic} is actually a system of $n$ decoupled linear equations in the $N$ variables $z\c{i}_k$, $i=1,\dotsc,N$, for each of $k=1,\dotsc,n$.
With some abuse of notation, we have for each component of $z\c{i}$
\begin{equation*}
  E_N z\c{\ast}_k = x_k.
\end{equation*}
Applying Lemma~\ref{lem:algebraic-solution}, since $E_N$ does not have zero eigenvalues, we conclude that this equation has a unique solution for each component $z\c{i}_k$, and thus the algebraic equation~\eqref{eq:limit-algebraic} has the unique solution given by \eqref{eq:singular-manifold}.

The slow model is then obtained by substituting \eqref{eq:singular-manifold} into $f(x,z)$, yielding
\begin{equation}
  \label{eq:slow-model}
  \dot x = f(x,h(x)) = F\bigl( \frac{x}{2} \bigr) - \degparam \frac{x}{2}.
\end{equation}

The boundary-layer model describes the dynamics of the deviation of the fast variables $z\c{i}$ from their slow-time-scale approximation \eqref{eq:singular-manifold}.
As described in Section~\ref{sec:bound-singular-perturbation}, the boundary-layer model is formulated in the variable
\begin{equation}
  \label{eq:fast-variables}
  y\c{i} = z\c{i} - \frac{x}{2}.
\end{equation}
Using the model equations \eqref{eq:popmodel-transformed}, the boundary-layer model is obtained as
\begin{equation}
  \label{eq:boundary-model}
  \frac{dy\c{i}}{d\tau} = - \frac{1}{V_c} (y\c{i} + \frac{1}{N} \sum_{j=1}^N y\c{j}), \qquad i=1,\dotsc,N.
\end{equation}

\begin{prop}
  \label{prop:population-dynamics}
  Assume that $\bar x$ is an equilibrium point of the slow system \eqref{eq:slow-model}, i.e., $F(\bar x/2) - k_d \bar x/2 = 0$.
  Let the following conditions be satisfied:
  \begin{enumerate}
  \item $F$ is globally Lipschitz, i.e., $\Vert F(\xi) - F(\xi^\prime) \Vert \leq L \Vert \xi - \xi^\prime \Vert$ for $\xi, \xi^\prime \in \Real^n_+$.
  \item Assumption 2 from above is satisfied for the function $f(x,h(x)) = F(x/2) - k_d x/2$ for $D_x = \Real^n_+$, where $\Vert x \Vert$ is replaced by $\Vert x - \bar x \Vert$ to account for the non-zero equilibrium point, and the Lyapunov function $V_S(x)$ satisfies
    \begin{equation}
      \label{eq:vs-extra-condition}
      \Vert \frac{\partial V_s}{\partial x}(x) \Vert \leq c \Vert x - \bar x \Vert
    \end{equation}
    for a constant $c > 0$.

  \end{enumerate}
  Define $w(t) = (\xi\c{i}(t)-\bar x/2)_{i=0,\dotsc,N}$, where $\xi\c{i}(t)$ is a solution of \eqref{eq:popmodel}.
  Let $x_0 \in \{ d_2 \Vert x - \bar x \Vert^2 \leq \varrho D \}$, with $0 \leq \varrho < 1$.
  Then there exists $\varepsilon^\ast > 0$ such that for all $\varepsilon < \varepsilon^\ast$, the solutions of the system~\eqref{eq:popmodel} are bounded by
  \begin{equation}
    \label{eq:bound-popmodel}
    \Vert w(t) \Vert \leq \beta(\Vert w(0) \Vert, t) + \mu(\varepsilon),
  \end{equation}
  and $\mu(\varepsilon) \rightarrow 0$ as $\varepsilon \rightarrow 0$.
\end{prop}
\begin{proof}
  The result is a direct consequence of Theorem~\ref{theorem:boundedness}.
  In the proof, it remains to verify that the transformed system~\eqref{eq:popmodel-transformed} satisfies the assumptions made for theorem \ref{theorem:boundedness}.
  
  For assumption 1, we have shown previously that the system~\eqref{eq:popmodel-transformed} has a quasi-steady state solution $z\c{i} = x/2$, $i=1,\dotsc,N$.
  The norm bound in assumption 1 is thus
  \begin{equation*}
    c_1 \geq \Vert \frac{\partial h}{\partial x}(x) \Vert = \frac{\sqrt{N}}{2}.
  \end{equation*}
  From the assumption in Proposition~\ref{prop:population-dynamics}, $(\bar x, h(\bar x))$ is an equilibrium point of the system~\eqref{eq:popmodel-transformed}.
  
  Assumption 2 from Section~\ref{sec:bound-singular-perturbation} is a condition in Proposition~\ref{prop:population-dynamics}, transformed to the non-zero equilibrium point.

  Assumption 3 concerns the boundary layer model~\eqref{eq:boundary-model}.
  This is a linear model, and with Lemma~\ref{lem:algebraic-solution} we can conclude that its dynamics matrix has $n$ eigenvalues at $-2/V_c$ and $n(N-1)$ eigenvalues at $-1/V_c$.
  It follows that there exists a Lyapunov function $V_B$ for the boundary layer model which satisfies assumption 3 in Section~\ref{sec:bound-singular-perturbation}.
  Also, since the boundary layer model does not depend on $x$, $V_B$ can be chosen such that $\partial V_B/\partial x = 0$.

  For assumption 4, we make use of the conditions that $F$ is globally Lipschitz and thus $f,g$ are globally Lipschitz in $x$, $z$, and $\varepsilon$, the derivative of $V_S(x)$ is bounded by~\eqref{eq:vs-extra-condition}, and $\partial V_B/\partial x = 0$.
  From these conditions, one can directly verify that there exist constants $c_4, \dotsc, c_{12}$ such that the inequalities~\eqref{eq:lyap3-8} are satisfied.

  The bound~\eqref{eq:bound-popmodel} is then a direct consequence of Theorem~\ref{theorem:boundedness}.
\end{proof}

\section{A POPULATION OF COUPLED SWITCHES}
\label{sec:simulation}

\subsection{Modeling and analysis}
\label{sec:switch-modeling-analysis}

The dynamics of a population of coupled cells become particularly interesting for intracellular dynamics with non-linear behavior such as bistability or oscillations.
Here, we focus on bistability, which has been thouroughly studied in the context of intracellular biological networks \cite{CherryAdl2000,FerrellXio2001,EissingWal2007}.
A typical model for an intracellular switch is given by the scalar differential equation \cite{WaldherrAll2010a}
\begin{equation}
  \label{eq:switch-ode}
  \dot \xi = F(\xi) = \frac{3\xi^2}{1+\xi^2} - \xi,
\end{equation}
where $\xi\in\Real$ denotes the expression level of an auto-activating gene.

We consider a system where the auto-activating gene product can be exchanged among cells.
Using the modeling approach presented in Section~\ref{sec:cp-dynamics}, the resulting population of coupled switches will be described by the system
\begin{equation}
  \label{eq:coupled-switches}
  \begin{aligned}
    \dot \xi\ext &= \sum_{j=1}^N \frac{\diffparam}{N V_c} (\xi\c{j} - \xi\ext) \\
    \dot \xi\c{i} &= \frac{3(\xi\c{i})^2}{1+(\xi\c{i})^2} - \xi\c{i} - \frac{\diffparam}{V_c} (\xi\c{i} - \xi\ext), \quad i=1,\dotsc,N,
  \end{aligned}
\end{equation}
where the degradation in the extra-cellular medium has been set to zero.

From the theoretical results in Section \ref{sec:cp-dynamics}, one can conclude that the population of coupled switches behaves as a single switch for sufficiently fast diffusive exchange: 
due to the fast coupling, the individual cells would quickly reach a state where all have approximatively the same concentration $\xi\c{1} \approx \xi\c{2} \approx \cdots \approx \xi\c{N}$, and from then on the dynamics would be governed by the slow bistable equation
\begin{equation}
  \label{eq:slow-switch-dynamics}
  \dot x = \frac{3 \frac{x^2}{4}}{1 + \frac{x^2}{4}} - \frac{x}{2}.
\end{equation}
The initial condition of the slow model is computed from the initial condition for the full model~\eqref{eq:switch-ode} by the coordinate transformation \eqref{eq:trafo}, yielding
\begin{equation}
  \label{eq:slow-switch-ic}
  x(0) = \xi\c{0}(0) + \frac{1}{N} \sum_{i=1}^N \xi\c{i}(0)
\end{equation}
The inverse transformation
\begin{equation}
  \label{eq:slow-switch-backtrafo}
  \xi\c{i}(t) = \frac{x(t)}{2}, \qquad i = 0, \dotsc, N
\end{equation}
can be used to obtain an approximate solution of the original population model~\eqref{eq:coupled-switches} from a solution of the slow model~\eqref{eq:slow-switch-dynamics}.

Motivated from these theoretical considerations, we perform a simulation study with a range of diffusion rate values and heterogeneous initial conditions for the switch population~\eqref{eq:coupled-switches}.
Simulation results for $N=10$ are shown in Figures~\ref{fig:sim-med-x0} and \ref{fig:sim-high-x0}.
The solutions for $\varepsilon = 0$ have been obtained by simulating the reduced model~\eqref{eq:slow-switch-dynamics} and using \eqref{eq:slow-switch-backtrafo} to compute the approximate solution for the full model.
By construction, this approximate solution has identical values for all individual switches.

As indicated by the theoretical considerations, we in fact observe switch synchronization and population level bistability for the coupled switch model~\eqref{eq:coupled-switches}: for sufficiently small $\varepsilon$, the states of individual switches quickly converge to each other, and then, depending on the initial conditions, converge all together to either the upper or the lower stable equilibrium point.

In contrast, for larger values of $\varepsilon$, individual switches may converge to different stable equilibrium points, depending on their individual initial condition.
An example for this is shown in Figure~\ref{fig:sim-high-x0} with $\varepsilon = 20$.

\begin{figure}
  \centering
  \includegraphics[width=7.2cm]{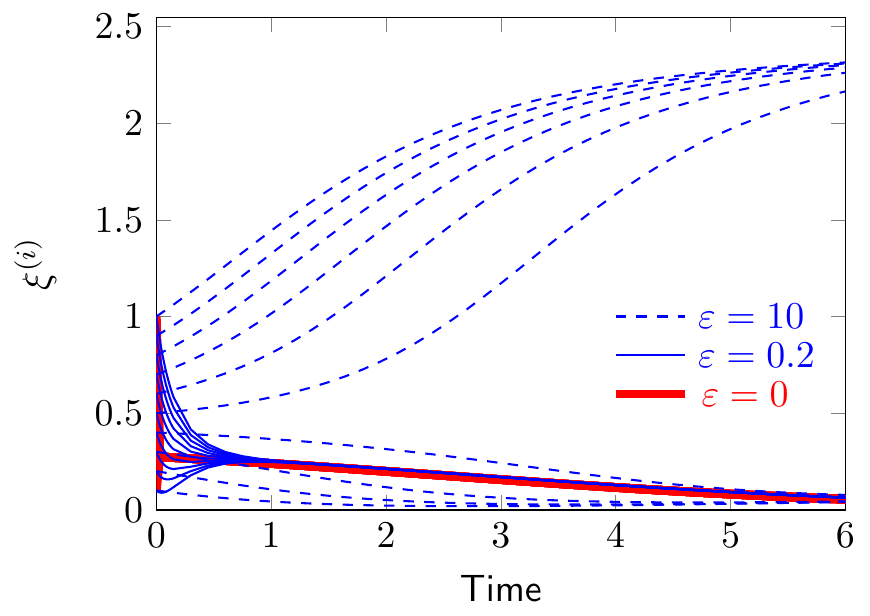}
  \caption{Simulation of the coupled switch population with $N=10$ for different values of $\varepsilon$ and initial conditions ranging from 0.1 to 1.}
  \label{fig:sim-med-x0}
\end{figure}

\begin{figure}
  \centering
  \includegraphics[width=7.2cm]{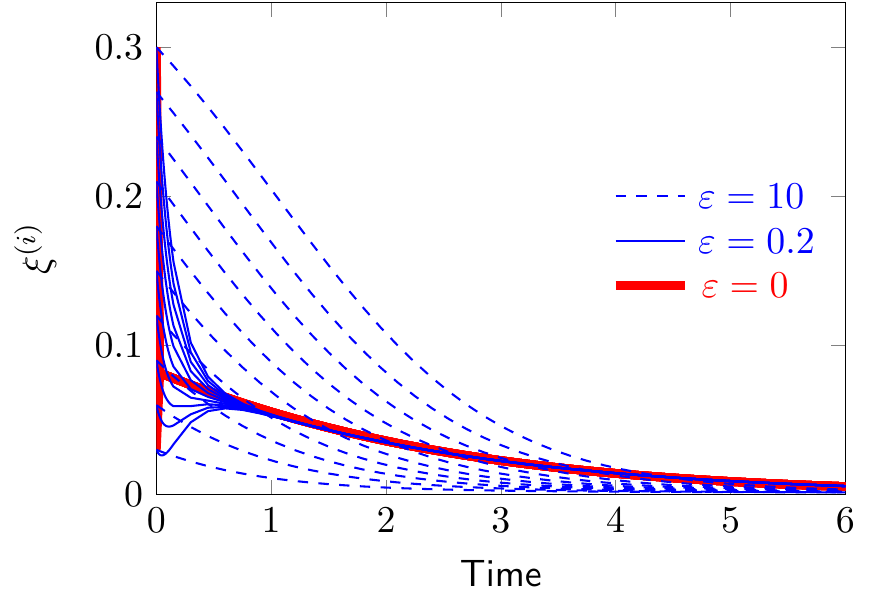}
  \caption{Simulation of the coupled switch population with $N=10$ for different values of $\varepsilon$ and initial conditions ranging from 0.03 to 0.3.}
  \label{fig:sim-lo-x0}
\end{figure}

\begin{figure}
  \centering
  \includegraphics[width=7.2cm]{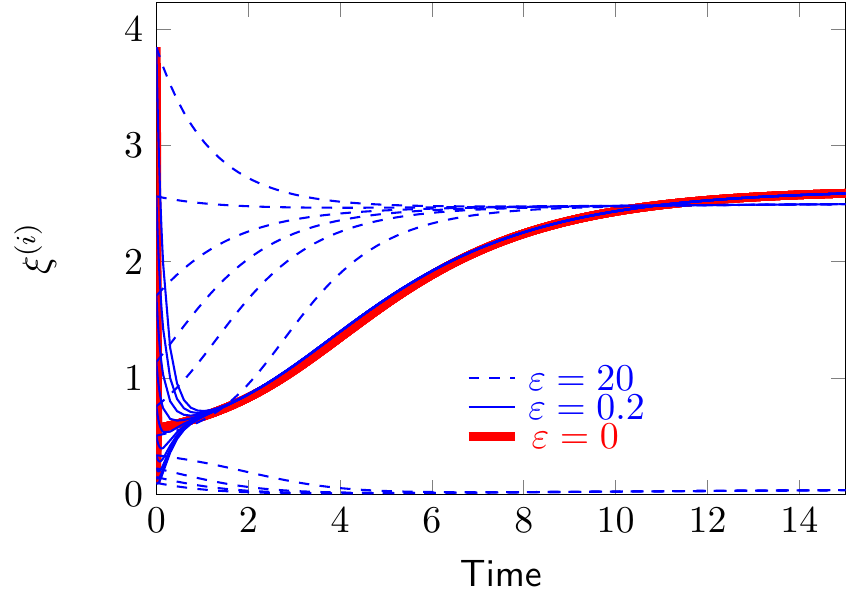}
  \caption{Simulation of the coupled switch population with $N=10$ for different values of $\varepsilon$ and initial conditions ranging from 0.1 to 4.}
  \label{fig:sim-high-x0}
\end{figure}

\subsection{Discussion of biological relevance}
\label{sec:switch-discussion}

As shortly discussed in the introduction, the relevance of bistable dynamics in gene regulation and biochemical signal transduction is very well established.
However, the relation between bistable intracellular dynamics and population level behavior of coupled cells has not been well studied so far.

Motivated from the theoretical results in Sections~\ref{sec:bound-singular-perturbation} and \ref{sec:cp-dynamics}, the simulation study described in this section indicates the occurence of population-level bistablity for fast enough diffusive exchange.
We speculate that this type of coupling may be a relevant biological mechanism to provide robustness of cellular decisions within a cell population under cellular heterogeneity and molecular noise.

Coordination among bistable switches in a cell population may for example ensure that a cluster of cells within a developing tissue differentiates together, even if not all cells receive a high enough stimulus to initiate single cell diferentiation.
Alternatively, this mechanism may be useful to extend engineered genetic switches on the single-cell level \cite{GardnerCan2000} to coordinated switching of a whole cell population using synthetic biology techniques.

\section{CONCLUSIONS}

We have proven convergence to a common equilibrium point for populations of cells which are coupled via diffusive exchange of molecules with a common extracellular medium, under the assumption of sufficiently fast diffusive exchange rates compared to the intracellular dynamics.
While we have considered only homogenous intracellular models in this paper, the results extend directly to heterogeneous cells.

A study of coupled bistable switches showed coordinated behavior and population-level bistability among the switches.
This has potential relevance for the robustness of tissue differentiation and for synthetic biology.



\section*{APPENDIX}

\begin{lemma}
  \label{lemma:B-negative}
  Let
  \begin{equation}
    \label{eq:B-def}
    B(x,y,\varepsilon) = k_1(\varepsilon) \normy - k_2 \normx^2 - \frac{k_3}{\varepsilon} \normy^2
  \end{equation}
  with $k_1$, $k_2$, $\varepsilon > 0$.
  If
  \begin{equation}
    \label{eq:xy-norm-bound}
    \normx^2 + \normy^2 \geq \frac{\varepsilon^2 (k_1(\varepsilon))^2}{k_3^2} + \frac{(k_1(\varepsilon))^2 \varepsilon}{4 k_2 k_3},
  \end{equation}
  then $B(x,y,\varepsilon) \leq 0$.
\end{lemma}

\begin{proof}
  \textbf{Case 1}
  Let $\normy^2 \geq \frac{\varepsilon^2 (k_1(\varepsilon))^2}{k_3^2}$.
  Then
  \begin{equation*}
      B(x,y,\varepsilon) \leq k_1(\varepsilon) \normy (1-\frac{k_3}{\varepsilon k_1(\varepsilon)} \normy) \leq 0.
  \end{equation*}
  \textbf{Case 2}
  Let $\normy^2 \leq \frac{\varepsilon^2 (k_1(\varepsilon))^2}{k_3^2}$.
  Then $\normx^2 \geq \frac{(k_1(\varepsilon))^2 \varepsilon}{4 k_2 k_3}$ and, by quadratic extension,
  \begin{equation*}
    \begin{aligned}
      B(x,y,\varepsilon) &= -k_2 \normx^2 - \frac{k_3}{\varepsilon} (\normy^2 - \frac{k_1(\varepsilon) \varepsilon}{2 k_3})^2 + \frac{(k_1(\varepsilon))^2 \varepsilon}{4 k_3} \\
      &\leq - k_2 \normx^2 + \frac{(k_1(\varepsilon))^2 \varepsilon}{4 k_3} \leq 0.
    \end{aligned}
  \end{equation*}
\end{proof}

\section*{ACKNOWLEDGMENT}

We thank Gerd Simon Schmidt and Jingbo Wu for helpful discussions of this research.


\bibliographystyle{plain}
\bibliography{/home/ist/waldherr/research/Referenzen.bib}

\end{document}